\documentclass[journal]{IEEEtran}

\usepackage{cite}
\usepackage{amsmath, amssymb, amsfonts, amsthm}
\usepackage{graphicx}
\usepackage{mathtools} 
\usepackage{enumerate}

\usepackage{xcolor}
\usepackage{tikz}
\usetikzlibrary{positioning, shapes, arrows, calc, fit}
\tikzstyle{vertex}=[circle, draw, inner sep=0pt, minimum size=12pt]
\newcommand{\vertex}{\node[vertex]}

\usepackage{pgfplots}
\pgfplotsset{compat=1.18}

\def\subparagraph{}
\usepackage{titlesec}
\titlespacing*{\section}{0pt}{2ex plus 1.5ex minus 0.2ex}{0.5ex plus 0.2ex}
\titlespacing*{\subsection}{0pt}{2ex plus 1ex minus 0.2ex}{0.5ex plus 0.2ex}
\newtheorem{lemma}{Lemma}
\newtheorem{definition}{Definition}
\newtheorem{theorem}{Theorem}
\newtheorem{assumption}{Assumption}
\newtheorem{remark}{Remark}

\allowdisplaybreaks

\newcommand{\R}{\ensuremath{\mathbb{R}}}	
\newcommand{\diag}{\operatorname{diag}}	

\newcommand{\sym}{\operatorname{sym}}
\newcommand{\col}{\operatorname{col}}
\newcommand{\rank}{\operatorname{rank}}
\def\T{\top}
\def\bmat{\begin{bmatrix}}
\def\emat{\end{bmatrix}}
\def\colsep{\setlength\arraycolsep}

\title{Distributed Prescribed-Time Observer for Nonlinear Systems in Block-Triangular Form}
\author{Vincent de Heij, M. Umar B. Niazi, Karl H. Johansson, and Saeed Ahmed
\thanks{
V.~de~Heij and Saeed Ahmed are with the Jan C. Willems Center for Systems and Control and the Engineering and Technology Institute Groningen (ENTEG), Faculty of Science and Engineering, University of Groningen, 9747~AG Groningen, The Netherlands. (Emails: v.de.heij@student.rug.nl, s.ahmed@rug.nl).}
\thanks{K.~H.~Johansson is with the Division of Decision and Control Systems, Digital Futures, KTH Royal Institute of Technology, SE-100 44 Stockholm, Sweden. (Email: kallej@kth.se)}
\thanks{M.~U.~B.~Niazi is with the Division of Decision and Control Systems, Digital Futures, KTH Royal Institute of Technology, SE-100 44 Stockholm, Sweden, and with the Laboratory for Information and Decision Systems, Department of Electrical Engineering and Computer Science, Massachusetts Institute of Technology, Cambridge, MA 02139, USA.
(Email: mubniazi@kth.se, niazi@mit.edu).}
\thanks{M.~U.~B.~Niazi is supported by the European Union’s Horizon Research and Innovation Programme under Marie Sk{\l}odowska-Curie grant agreement No. 101062523.}
}

\begin{document}

\maketitle

\begin{abstract}
This paper proposes a distributed prescribed-time observer for nonlinear systems representable in a block-triangular observable canonical form. Using a weighted average of neighbor estimates exchanged over a strongly connected digraph, each observer estimates the system state despite the limited observability of local sensor measurements. The proposed design guarantees that distributed state estimation errors converge to zero at a user-specified convergence time, irrespective of observers' initial conditions. To achieve this prescribed-time convergence, distributed observers implement time-varying local output injection gains that monotonically increase and approach infinity at the prescribed time. The theoretical convergence is rigorously proven and validated through numerical simulations, where some implementation issues due to increasing gains have also been clarified.
\end{abstract}


\section{Introduction}
Distributed state estimation involves estimating the state of a dynamical system through a network of spatially distributed observers, where the sensor of each observer can measure only a partial state. Due to the limited observability of local sensors, individual observers cannot independently reconstruct the system state. Therefore, observers employ information exchanged by their neighbors over a communication network to achieve full state estimation, as illustrated in Fig.~\ref{fig:disobsnetwork}. Such distributed approaches have emerged as a promising alternative to centralized methods for state estimation, offering significant advantages in scalability, resilience, privacy, and communication efficiency \cite{Tajer2010}.

For linear systems, seminal papers \cite{Olfati2007} and \cite{Khan2008} studied distributed Kalman filtering problem, motivating developments of distributed Luenberger observers \cite{Kim2016, Park2017, Wang2018, Trentelman19} in deterministic settings. Extensions to finite-time distributed observers, such as \cite{Silm2020,Ortega2021}, aim to estimate the state of a linear system up to a certain accuracy within a finite time, which depends on the system's initial condition. A more recent approach in \cite{Ge2023} proposes a kernel-based method for distributed state estimation to achieve fixed-time convergence of the state estimate for a linear system, where the convergence time is independent of the system's initial condition.

Distributed observers for nonlinear systems represent a significant advancement in the field. A distributed Luenberger-like observer asymptotically estimating the state of autonomous nonlinear systems is proposed in \cite{Battilotti2019}. For controlled nonlinear systems, \cite{Yuanqing22} proposes a distributed observer by adapting methods from \cite{Kim2016} and \cite{Trentelman19}. 
However, these methods ensure only asymptotic convergence of the estimate to the system state.

 \begin{figure}[t]
      \centering
      \includegraphics[width=0.95\linewidth]{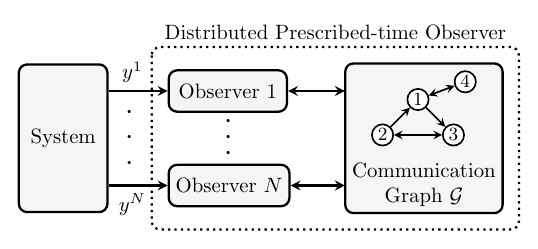} 
      \caption{Framework of distributed state estimation.}
      \label{fig:disobsnetwork}
  \end{figure}


In this paper, we are interested in distributed observers that can estimate the state within finite time, which is prescribed by the user and is independent of the system's initial condition. This is because in practical applications, especially those with uncertain initial conditions, achieving an accurate state estimate within a prescribed time is critical to meet control objectives. Existing methods for finite-time state estimation often have convergence times that depend on the system's initial conditions, which may be unavailable, inaccurate, or confidential in practical applications; whereas methods for fixed-time state estimation have convergence times that depend on observer parameters (see \cite{Ge2023}) and cannot be chosen arbitrarily by the user. On the other hand, the prescribed-time state estimation offers a compelling alternative, as it allows the user to \textit{a priori} specify the convergence time independently of the initial condition.




Centralized prescribed-time observers for linear and nonlinear systems in triangular forms are proposed in \cite{Holloway19} and \cite{Adil2024}, respectively, which employ time-varying gains that approach infinity at the prescribed time. To the best of our knowledge, the problem of distributed prescribed-time state estimation of nonlinear systems remains unresolved in current literature. While \cite{Chang2021} attempted to address this problem, the paper contains critical mathematical errors as pointed out in Appendix~\ref{append:mistakes-chang2021}.




In this paper, we address the distributed prescribed-time state estimation problem for nonlinear systems in a block-triangular observable canonical form \cite{Gauthier2001, Shim2001, Hammouri2010}.
Systems in this form can be obtained by transforming nonlinear systems via a specific diffeomorphic map subject to the differential observability condition \cite{Shim2001, gauthier1992}.
Our proposed distributed observer allows users to specify an arbitrary convergence time for state estimation, independent of initial conditions. Because of the limited observability of local measurements, observers communicate over a strongly connected directed graph to estimate locally observable states through output injection and unobservable states through consensus with their neighbors. Using time-varying gains that approach infinity near convergence time and synthesizing observer gains via linear matrix inequalities, we prove convergence of the distributed state estimation errors to zero within the prescribed time. We validate our results through numerical simulations and also provide guidelines on the implementation issues due to increasing time-varying gains.

\section{Preliminaries}
\label{sec:prelim}
\textit{Notations:}
Given $A\in\mathbb{R}^{n\times n}$, $\sym(A):= A+A^\T$. The Kronecker product between $A$ and $B$ is given by $A\otimes B$. The identity matrix of dimensions $n\times n$ is $I_n$, a vector of ones is $\mathbf{1}_n\in\mathbb R^{n\times 1}$, and a zero matrix is $0_{m\times n}\in\mathbb R^{m\times n}$. A symmetric matrix $P=P^\T\in\mathbb R^{n\times n}$ is denoted as $P\geq 0$ (resp., $P\leq 0$) if it is positive (resp., negative) semi-definite.
For matrices $A_1,\dots,A_N$, $\underline A = \diag(A_1,\ldots,A_N)$ denotes a block diagonal matrix with $A_i$'s as diagonal blocks.
A block diagonal matrix comprising $N$ copies of $A$ is denoted by $\overline A = I_N\otimes A$. 
The vertical concatenation of $v_1\in\mathbb R^{n_1},\ldots,v_N\in\mathbb R^{n_N}$ into a single column vector is denoted by $\col(v_1,\dots,v_N)$.


\textit{Graph theory:}
A weighted directed graph is denoted by $\mathcal{G} = (\mathcal{N}, \mathcal{E}, \mathcal{A})$ with $\mathcal{N} = \{1, 2, \ldots, N\}$ the set of nodes, $\mathcal{E} \subset \mathcal{N} \times \mathcal{N}$ the set of edges, and $\mathcal{A} = [a_{ij}] \in \mathbb{R}^{N \times N}$ the weighted adjacency matrix. The edge $(i,j)\in\mathcal E$ indicates that node~$i$ receives information from node~$j$. Each entry $a_{ij}$ in $\mathcal{A}$ denotes the weight associated with the edge $(i,j)$, with $a_{ij}>0$ if and only if $(i,j) \in \mathcal{E}$; otherwise, $a_{ij} = 0$. 
A directed path from node~$j$ to node~$i$ consists of a sequence of edges $(i, i_{1}),\dots,(i_k,j)$, for some integer $k\geq 1$. 
A directed graph $\mathcal{G}$ is called \textit{strongly connected} if, for any $i,j\in\mathcal N$, there exists a directed path from $j$ to $i$ in $\mathcal E$. 

Given $\mathcal G$ is strongly connected, then (see \cite{Trentelman19, Wenwu2010, Olfati2004, Wei2005}):

\begin{enumerate}[{Fact} 1:]
    \item The Laplacian matrix $\mathcal{L}\in\mathbb R^{N\times N}$ of $\mathcal G$ has a zero eigenvalue with the corresponding eigenvector $\mathbf{1}_N$ and all other eigenvalues of $\mathcal{L}$ lie in $\mathbb C_{>0}$.
    \label{fact1}
    \item There exists a unique $r = \bmat r_1 & \ldots & r_N\emat\in\mathbb R_{>0}^{1\times N}$ such that $r \mathcal{L} = 0_{1\times N}$ and $r\textbf{1}_N = N$. \label{lem:stronggraphs}
    \item $\hat {\mathcal{L}}\coloneqq R\mathcal{L}+\mathcal{L}^\T R$ is positive semi-definite, where
    \begin{equation} \label{eq:R-matrix}
        R = \diag(r_1,\ldots,r_N).
    \end{equation}
    Moreover, $\textbf{1}^{\T}_N \hat {\mathcal{L}} = 0$ and $\hat {\mathcal{L}} \textbf{1}_N = 0.$
    \label{fact3}
\end{enumerate}

\section{System Definition}
\label{sec:sys-def}
Consider a nonlinear system
\begin{align}
\begin{aligned}\label{eq:sys}
        \dot x(t) &= Ax(t) + B\varphi(x(t))\\
        y^i(t) &= H^i x(t), \quad i=1,\dots,N
\end{aligned}
\end{align}
where $x(t)\in\mathbb R^n$ is the state, $y^i(t)\in\mathbb R$ is sensor~$i$'s measurement, and $\varphi:\mathbb R^n\to\mathbb R^N$ is a known function.
Specifically, the system comprises $N\leq n$ subsystems, where the state of subsystem~$i$ is $x_i(t)\in\mathbb R^{n_i}$ with $\sum_{i=1}^N n_i = n$. Therefore, $x(t)= \col(x_1(t), \ldots, x_N(t))$ and
\begin{equation*}
    x_i(t) = \col(x_{i1}(t),\ldots,x_{i,n_i}(t)), \quad i=1,\ldots,N.
\end{equation*}
The nonlinearity $\varphi$ is triangular, meaning
\begin{equation}
\label{eq:var-phi}
\varphi(x)= \col(\varphi_1(x_1),\varphi_2(x_1,x_2),\ldots,\varphi_N(x_1\ldots,x_N))
\end{equation}
where $\varphi_i:\mathbb R^{\bar n_i}\to\mathbb R$ with $\bar n_i=\sum_{j=1}^i n_j$.

Let $H = \col(H^{1},\ldots,H^{N})$, where $H^i \in \mathbb{R}^{1 \times n}$ given by
\begin{equation*}
H^i=\colsep{2pt}\begin{bmatrix}
    0_{1\times n_1} & \dots & 0_{1\times n_{i-1}} & H_i & 0_{1\times n_{i+1}} & \dots & 0_{1\times n_N}
\end{bmatrix}
\end{equation*}
with $H_i\in\mathbb R^{1\times n_i}$ placed on the $i$-th block position. Notice that 
$y^i(t) = H^i x(t) = H_i x_i(t)$. The system matrices ${A\in\mathbb{R}^{n\times n}}$, $B\in\mathbb{R}^{n\times N}$, and $H\in\mathbb{R}^{N\times n}$ are given by
\begin{align*}
    A& = \diag(A_1,\ldots,A_N), \quad B = \diag(B_1,\ldots,B_N), \\
    H& = \diag(H_1, \ldots, H_N)
\end{align*}
with $A_i \in \mathbb{R}^{n_i \times n_i}$, $B_i \in \mathbb{R}^{n_i \times 1}$, and $H_i \in \mathbb{R}^{1 \times n_i}$ in observable canonical form
\begin{align*}
 A_i &=
    \begin{bmatrix}
        0_{(n_i-1)\times 1}  & I_{n_i-1}\\
        0 & 0_{1\times (n_i-1)}
    \end{bmatrix}, \quad B_i =
    \begin{bmatrix}
    0_{(n_i-1 )\times 1}\\
    1
    \end{bmatrix}, \\
    H_i &=
    \colsep{10pt}\begin{bmatrix}
        1 & 0_{1\times(n_i-1)}
    \end{bmatrix}.
\end{align*}

\begin{assumption}\label{as:1}
The function $\varphi_i:\mathbb R^{\bar n_i}\to\mathbb R$, for $i=1,\dots,N$, in \eqref{eq:var-phi} is Lipschitz continuous, i.e., there exists $\gamma_i\in\mathbb R_{\geq 0}$ such that for all $\bar x_i,\bar \omega_i \in \mathbb{R}^{\bar n_i}$, it holds that
\begin{equation}\label{eq:Lipschitz}
    \bigl|\varphi_i(\bar x_i + \bar\omega_i) - \varphi_i(\bar x_i)\bigr| \leq \gamma_i\|\bar\omega_i\|
\end{equation}
where $\bar x_i \coloneqq \col(x_1,\dots,x_i)$ and $\bar \omega_i \coloneqq \col(\omega_1,\dots,\omega_i)$. 
\end{assumption}

The communication network between the sensors is defined by a weighted directed graph $\mathcal G=(\mathcal N,\mathcal E,\mathcal A)$.

\begin{assumption} \label{as:strongcon}
    $\mathcal{G}$ is strongly connected. 
\end{assumption}

We further assume that the linear part of the system is jointly observable. However, it must be noted that the class of systems defined above becomes unobservable if any of the sensors are removed. Therefore, all measurements $y^i(t)$, for $i=1,\dots,N$, must be utilized for state estimation (cf. \cite{Yuanqing22}). 


\begin{assumption} \label{as:jointobs}
   $\rank[H^\T ~ (HA)^\T ~ \dots ~ (HA^{n-1})^\T] = n$.
\end{assumption}

Subject to the assumptions above, our objective is to design a distributed estimation algorithm with $N$ observers to cooperatively determine the full state of $x(t)$ of system \eqref{eq:sys} within a finite time $T>0$, which is prescribed by the user arbitrarily and independent of the initial conditions.

\section{Design Problem}

The distributed state estimation algorithm comprises $N$ observers, where each observer~$i$ can access only the measurement of sensor~$i$ and a weighted combination of state estimates of its neighboring observers according to the communication network $\mathcal G$.
Let $z^i(t)\in\mathbb R^n$ denote the state estimate of observer~$i$.
Each observer~$i$ estimates the locally observable part of the state, $x_i(t)$, by using the local output injection $y^i(t)-H^iz^i(t)$. 
The unobservable part of the state, $\col(x_1(t),\dots,x_{i-1}(t),x_{i+1}(t),\dots,x_N(t))$, is estimated by seeking consensus with the neighboring observers, i.e., $\sum_{j} a_{ij} [z^i(t)-z^j(t)]$. 
In particular, we consider the following architecture of the distributed observer
\begin{multline}\label{eq:localobserver}
    \dot z^i(t) = A z^i(t) + B\varphi(z^{i}(t)) +{\Gamma}(t) L^i[y^i(t) - H^iz^i(t)]\\
    - \mu^{1+m}(t)r_i   \sum^{N}_{j=1} a_{ij}[z^{i}(t)-z^j(t))], ~ i=1,\dots,N
\end{multline}
where $z^i(t) \in \mathbb{R}^n$ denotes the full-state estimate of observer~$i$. 
The entry $a_{ij}$ is the $(i,j)$-th entry of the adjacency matrix $\mathcal{A}$ of $\mathcal{G}$ and $r_i$ is defined in Fact~\ref{lem:stronggraphs} in Section~\ref{sec:prelim}. 
The distributed observer design problem aims to find gains $\Gamma(t)\in\mathbb R_{>0}^{n\times n}$, $\mu(t)\in\mathbb R_{>0}$, and $L^i\in\mathbb R^n$ such that a certain stability property (to be defined later) of the state estimation error holds.

\subsection{Coordinate Transformation and Time-Varying Gains}\label{sub:designapproach}
The state estimation error of observer~$i$ in \eqref{eq:localobserver} is defined as 
\[
e^i(t) = x(t) - z^i(t), \quad i=1,\dots,N.
\]
Consider a time-varying coordinate transformation
\[
\zeta^{i}(t) = \Gamma^{-1}(t) e^{i}(t)
\]
where $\Gamma(t)\in\mathbb{R}_{>0}^{n\times n}$ is expressed as
\begin{equation}\label{eq:Gamma}
\Gamma(t) \coloneqq \Gamma(\mu(t)) = \diag\left(
\Gamma_1(\mu(t)),\dots, \Gamma_N(\mu(t))
\right)
\end{equation}
with $\Gamma_i(\mu(t)) \in \mathbb{R}^{n_i \times n_i}$, associated with subsystem~$i$,
\begin{equation} \label{eq:Gamma-i}
    \Gamma_i(t) \coloneqq \Gamma_i(\mu(t)) = \diag\left(
     \mu^{1+m}(t),\ldots,\mu^{n_i(1+m)}(t)
     \right).
\end{equation}
Here, $m\geq 1$ and $\mu:[t_0,t_0+T) \to \mathbb R_{>0}$ is defined as
\begin{equation}\label{eq:mu}
        \mu(t)\coloneqq\mu(t;t_0,T)= \frac{T}{T+t_0-t}
\end{equation}
for some initial time $t_0\in\mathbb R_{\geq 0}$ and prescribed-time $T\in\mathbb R_{>0}$. Notice that $\mu$ is parametrized by $t_0$ and $T$.
In \eqref{eq:Gamma-i}, $m$ is a design parameter that controls the convergence speed, and $\mu(t)\in[1,\infty)$ is a monotonically increasing function approaching infinity as $t\to t_0 + T$. Consequently, its multiplicative inverse $1/\mu(t)\in(0,1]$ decreases monotonically to $0$ as $t\to t_0+T$. 

\subsection{Prescribed-Time Convergence of Estimation Error}
Let the vector of the state estimates of $N$ observers be
$
z(t) = \col(z^{1}(t),\ldots,z^N(t))
$
and the corresponding distributed state estimation error
\begin{equation}\label{eq:distributederror}
 e (t)= [\mathbf{1}_N \otimes x(t)] -z(t). 
\end{equation}
By taking the derivative of \eqref{eq:distributederror} and combining \eqref{eq:sys} and \eqref{eq:localobserver}, we obtain the error system as
\begin{equation}\label{eq:errorthm}
    \dot e(t) = \Lambda(t) e(t) + \overline B \Delta \Phi(t) - \mu^{1+m}(t) (R\mathcal{L}\otimes I_n) e(t)
\end{equation}
where $\overline B  = I_N \otimes B$, $\mathcal L$ is the Laplacian matrix, $R$ in \eqref{eq:R-matrix}, and
\begin{align*}
    \Delta \Phi(t) \!&=\col\!\left(\varphi(x(t))- \varphi(z^1(t)),\ldots,\varphi(x(t))- \varphi(z^N(t))\right), \\
    \Lambda(t) &= \diag(A-\Gamma(t) L^1H^1,\ldots,A-\Gamma(t) L^NH^N).
\end{align*}
Now, consider the transformed error
\begin{equation}\label{eq:distributederrortransformation}
    \zeta(t) = \overline{\Gamma}^{-1}(t) e(t)
\end{equation}
where $\overline{\Gamma}(t) \coloneqq I_N \otimes \Gamma(t)$. 
Taking its derivative yields
\begin{multline*}
        \dot \zeta(t) = \dot {\overline \Gamma}^{-1}(t)e(t) + {\overline \Gamma^{-1}(t) }\Lambda e(t) +{ \overline\Gamma^{-1}(t) \overline B} \Delta \Phi(t) \\
        -\mu^{1+m}(t) { \overline \Gamma^{-1}}(t)(R\mathcal{L}\otimes I_n)e(t).
\end{multline*}
Let $D_i = \diag(1,2,\ldots,n_i)$ and
\begin{equation}\label{eq:matrices-ol-ul}
    \arraycolsep=2pt\begin{array}{ll}
    \overline A \!=\! I_N \otimes A, &
    \overline D \!=\! I_N \otimes \diag(D_1,\ldots,D_N), \\
    \underline L \!=\! \diag(L^{1},\ldots,L^{N}), &
    \underline H \!=\! \diag(H^{1},\ldots,H^{N}).
    \end{array}
\end{equation}
Since $e(t) = {\overline{\Gamma}}(t) \zeta(t)$, and using the facts 
\begin{align*}
    & \dot{\overline\Gamma}^{-1}(t) = -\overline{\Gamma}^{-2}(t) \dot{\overline \Gamma}(t) = (1+m)\frac{\mu(t)}{T}\overline D\, \overline \Gamma^{-1}(t), \\
    & \overline\Gamma^{-1}(t) \overline A\, \overline\Gamma(t) = \mu^{1+m}(t) \overline A
\end{align*}
we obtain
\begin{multline}
    \label{eq:transformederror}
    \dot \zeta(t) =\mu^{1+m}(t)(\overline{A}-\underline{LH}- (R\mathcal{L}\otimes I_n))\zeta(t)\\
    \quad-(1+m)\frac{ \mu(t)}{T} \overline{D}\zeta(t)+{\overline \Gamma^{-1}}(t) \overline B \Delta\Phi(t).
\end{multline}



\begin{definition}[FGUAS \cite{Holloway19}]
  \label{def:FTGAS1}
A system $\dot e = f(e,t)$ is \emph{fixed-time, globally uniformly asymptotically stable} over an interval $[t_0,t_0+T)$ with $T\in\mathbb R_{>0}$ if there exists a class $\mathcal{KL}$ function\footnote{A function $\beta:\mathbb R_{\geq 0}\times\mathbb R_{\geq 0}\to\mathbb R_{\geq 0}$ is of class $\mathcal{KL}$ if (i)~for each fixed $s\in\mathbb R_{\geq 0}$, $\beta(\cdot,s)$ is continuous, zero at zero, and strictly increasing; and (ii)~for each fixed $r\in\mathbb R_{\geq 0}$, $\beta(r,s)\to 0$ as $s\to\infty$.} $\beta:\mathbb R_{\geq 0}\times\mathbb R_{\geq 0}\to\mathbb R_{\geq 0}$ such that
\begin{equation}
    \label{eq:KL-FTGUAS}
    \| e(t) \| \leq \beta (\|e({t_0})\|,\mu(t;t_0,T)-1), ~\forall t\in [t_0,t_0+T)
\end{equation}
where $\mu(t;t_0,T)$ is defined in \eqref{eq:mu}.
\end{definition}


Similar to \cite{Holloway19}, we employ the definition of fixed-time stability, which guarantees convergence within finite time $T$ that is independent of initial conditions. 
Prescribed-time observers build upon this notion but via specific time-varying gain \eqref{eq:Gamma} transform the fixed-time property into prescribed-time behavior--- i.e., for any arbitrary $T$ prescribed by the user, convergence is ensured at $T$ rather than just before some constant upper bound.
Therefore, when a user can arbitrarily select the fixed time $T$ and the system continues to satisfy the FGUAS property, we call the system to be \textit{prescribed-time, globally uniformly asymptotically stable}.

Our goal is to reconstruct the state before time $t_0+T$. To this end, it is restrictive to satisfy \eqref{eq:KL-FTGUAS} over the whole interval $[t_0,t_0+T)$. Therefore, inspired by the eventual stability notion \cite{lasalle1963}, we can relax Definition~\ref{def:FTGAS1} as follows.

\begin{definition}[E-FGUAS]
\label{def:eFTGUAS}
A system $\dot e = f(e,t)$ is \emph{eventually fixed-time, globally uniformly asymptotically stable} over an interval $[t_0,t_0+T)$ if there exists $t_0\leq t_*<t_0+T$ such that $\|e(t)\|<\infty$, for all $t\in[t_0,t_*)$, and \eqref{eq:KL-FTGUAS} holds for all $t\in[t_*,t_0+T)$.
\end{definition}

Using Definition~\ref{def:eFTGUAS}, we aim to develop a distributed observer design criterion to ensure that error system \eqref{eq:errorthm} converges to zero as $t\to t_0+T$.
Since the time $T$ is prescribed, one can substitute ``prescribed-time'' for ``fixed-time'' in the E-FGUAS property.
However, showing stability within time $T$ for time-varying systems \eqref{eq:errorthm} and \eqref{eq:transformederror} constitutes a stability analysis problem rather than a stabilization problem.
Consequently, similar to \cite{Holloway19}, we utilize fixed-time stability terminology and avoid using ``prescribed-time stability'' in our analysis, though it should be understood implicitly that an observer corresponding to an E-FGUAS error system achieves state reconstruction within the prescribed time $T$.


\section{Main Result}

This section provides a design criterion for the distributed prescribed-time observer \eqref{eq:localobserver} that achieves the E-FGUAS property of error dynamics \eqref{eq:errorthm}.
Given the time-varying gains $\Gamma(t)$ and $\mu(t)$ in \eqref{eq:Gamma} and \eqref{eq:mu}, one can synthesize local output injection gains $L^i$, for $i=1,\dots,N$, using this criterion.



\begin{theorem}\label{thm:newthm}
Given an initial time $t_0\in\mathbb R_{\geq 0}$, the prescribed-time $T\in\mathbb R_{>0}$, and $\mu(t;t_0,T)$ as in \eqref{eq:mu}, suppose there exist symmetric positive definite matrices $P_1,\dots,P_N\in\mathbb R^{n\times n}$, row vectors $Q_1,\dots,Q_N\in\mathbb R^{1\times n}$, scalars $\epsilon_1,\epsilon_2,\epsilon_3\in\mathbb R_{>0}$, and $\mu_\ast \coloneq \mu(t_\ast)$, for some $t_\ast\in[t_0,t_0+T)$, such that
\begin{subequations}\label{eq:LMISnew}
    \begin{align}
        \label{eq:LMI-a}
        & \sym\left(P\overline A-Q^{\T}\underline H - P (R\mathcal{L}\otimes I_n)\right)+ \epsilon_1 I_{nN} \leq 0 \\
        \label{eq:LMI-b}
        & \sym(P \overline D) - \epsilon_2 I_{nN} \geq 0 \\
        \label{eq:LMI-c}
        & P-\epsilon_3 I_{nN} \leq 0\\
        \label{eq:LMI-d}
        & \epsilon_2 -\frac{2\epsilon_3k_fT}{\mu_\ast(1+m)} \geq 0
    \end{align}
\end{subequations}
where $P=\diag(P_1,\dots,P_N)$, $Q=\diag(Q_1,\dots,Q_N)$, $R$ is given in \eqref{eq:R-matrix}, $k_f \!=\! \sqrt{N \sum_{i=1}^N (n_i\gamma_i)^2}$ with $\gamma_i$ in \eqref{eq:Lipschitz}, and $\overline A, \underline H, \underline L, \overline D$ in \eqref{eq:matrices-ol-ul}.
Then, the error system \eqref{eq:errorthm} is E-FGUAS over the interval $[t_0,t_0+T)$ with the observer gains given by
\begin{equation}\label{eq:Lsnew}
    L^{i} = P_i^{-1} Q_i^{\T}, \quad i=1,\dots,N.
\end{equation}
\end{theorem}

Before proving the theorem, we provide some remarks regarding the feasibility of the LMIs \eqref{eq:LMISnew}.
Firstly, the feasibility of \eqref{eq:LMI-a} is guaranteed by the detectability of the pair $(\underline H, \overline A-(R\mathcal L\otimes I_n))$, which holds when Assumption~\ref{as:strongcon} and Assumption~\ref{as:jointobs} are satisfied \cite{ugrinovskii2013}. 
This implies the existence of symmetric positive definite matrices $P_1,\dots,P_N$ satisfying \eqref{eq:LMI-a}. However, additional requirements \eqref{eq:LMI-b}--\eqref{eq:LMI-d} restrict the search space of feasible solutions $P$, but they are necessary for establishing the E-FGUAS property.

Due to the required prescribed-time convergence, the LMIs \eqref{eq:LMISnew} may become infeasible at $t_0$, where $\mu(t_0)=1$. Allowing a larger threshold $\mu_\ast \ge \mu(t_0)$ reduces conservatism in $\epsilon_2$ and ensures feasibility. Since $\mu(t)$ increases to $+\infty$ on $[t_0,t_0+T)$, it necessarily reaches $\mu_\ast$ at some $t_\ast \in [t_0,\,t_0+T)$, so the resulting design ensures the observer error converges to zero strictly before $t=t_0+T$.
This is the reason why it is generally impossible, as also observed in \cite{Holloway19}, to establish the FGUAS property of the error system \eqref{eq:errorthm} over the interval $[t_0,t_0+T)$. However, one can show that there exists $t_\ast\in[t_0,t_0+T)$ such that the system exhibits E-FGUAS property, ensuring the error $e(t)\to 0$ as $t\to t_0+T$. 

\begin{lemma} \label{lem:key-ineq}
    Let Assumption~\ref{as:1} hold. Then, 
    \begin{equation} \label{eq:k_f-ineq}
        \|\overline\Gamma^{-1}(t) \overline B \Delta\Phi(t)\| \leq k_f \|\zeta(t)\|
    \end{equation}
    where $k_f \!=\! \sqrt{\!N\! \sum_{i=1}^N \!(n_i\gamma_i\!)^2}$ with $\gamma_i$ in \eqref{eq:Lipschitz},  $\overline B$ and $\Delta\Phi(t)$ defined after \eqref{eq:errorthm}, and $\overline\Gamma^{-1}\!(t)\!\coloneqq\!\overline\Gamma^{-1}\!(\mu(t))\!$ defined after \eqref{eq:distributederrortransformation}.
\end{lemma}
\begin{proof}
    See Appendix~\ref{append:tech-lemmas}.
\end{proof}

\begin{proof}[Proof of Theorem~\ref{thm:newthm}]
We perform Lyapunov analysis to show the E-FGUAS property of the transformed error system \eqref{eq:transformederror} and then establish the E-FGUAS property of the original error system \eqref{eq:errorthm}. 
The proof is constructive in the sense that we derive an explicit upper bound of the distributed state estimation error that converges to zero as $t\to t_0+T$.

Let $P=\diag(P_1,\dots,P_N)$ with $P_i\in\mathbb R^{n\times n}$ some positive definite matrices. 
Consider a candidate Lyapunov function
    \begin{equation*}
        V(\zeta(t)) = \zeta(t)^{\T} P\zeta(t).
    \end{equation*}
Define $\Theta\coloneqq\overline{A}-\underline{LH}-(R\mathcal{L}\otimes I_n)$.
Then,
\begin{multline}
    \label{eq:V-dot}
    \dot V(\zeta(t)) =\ \mu^{1+m}(t) \zeta(t)^{\T} ( \Theta^{\T} {P} + {P} \Theta ) \zeta(t) \\
    \qquad\qquad - (1+m) \frac{{\mu(t)}}{T}\zeta(t)^{\T}( \overline{D}^{\T} {P}+P\overline{D}) \zeta(t) \\
    + 2 \zeta(t)^{\T} {P} \overline{\Gamma}^{-1}(t)\overline B\Delta\Phi(t).
\end{multline}
From Lemma~\ref{lem:key-ineq}, we have
\begin{align*}
        2 \zeta(t)^{\T} {P} \overline{\Gamma}^{-1}(t)\overline B \Delta \Phi(t) &\leq 2 \lambda_{\max}(P) \|\zeta(t)\| \|\overline{\Gamma}^{-1}(t)\overline B \Delta \Phi(t)\| \\
        &\leq 2k_f\lambda_{\max}(P) \| \zeta(t) \|^2
\end{align*}
and using \eqref{eq:LMI-a}, \eqref{eq:LMI-b}, and \eqref{eq:Lsnew}, we can upper bound \eqref{eq:V-dot} as
\begin{multline}
\label{eq:V-dot-bound}
    \dot V(\zeta(t)) \leq \bigg( -\epsilon_1 \mu^{1+m}(t) -\epsilon_2 \frac{(1+m)\mu(t)}{T} \\
    +2k_f \lambda_{\max}(P)\bigg) \|\zeta(t) \|^2.
\end{multline}
Firstly, to establish that $\|\zeta(t)\|$ remains bounded on $[t_0,t_\ast)$, we can further upper bound \eqref{eq:V-dot-bound} as
\[
\dot V(\zeta(t)) \leq 2k_f \lambda_{\max}(P)\|\zeta(t)\|^2.
\]
Since $V(\zeta) > 0$ for all $\zeta \in \R^{nN}  \setminus \{0\}$, by integrating over $[t_0,t)$, for $t_0\leq t < t_\ast$, we obtain
\begin{equation*}
V(\zeta(t)) \leq V(\zeta(t_0))\exp\left(\frac{2k_f \lambda_{\max}(P)(t-t_0)}{\lambda_{\min}(P)}\right).
\end{equation*}

Secondly, we show $\zeta(t)$ satisfies \eqref{eq:KL-FTGUAS} for $t\in[t_\ast, t_0+T)$. From \eqref{eq:LMI-c}, we have $\lambda_{\max}(P) \leq \epsilon_3$. Moreover, \eqref{eq:LMI-d} guarantees that on the interval $[t_\ast,t_0+T)$,
\[
\epsilon_2\frac{(1+m)\mu(t)}{T}\geq 2k_f \epsilon_3 \geq 2k_f \lambda_{\max}(P).
\]
Therefore, for $t\in[t_\ast,t_0+T)$, \eqref{eq:LMISnew} implies
\begin{equation*}
    \dot{V}(\zeta(t)) \leq -\epsilon_1 \mu^{1+m}(t)\|\zeta(t) \|^2.
\end{equation*}
Since $V(\zeta(t))\leq \lambda_{\max}(P) \|\zeta(t)\|^2$,
we have
\begin{equation*}
\frac{\dot{V}(\zeta(t))}{V(\zeta(t))} \leq - \frac{\epsilon_1 \mu^{1+m}(t)}{\lambda_{\max}(P)}, ~\text{for}~ \zeta(t)\neq 0.
\end{equation*}
Integrating over $[t_\ast,t)$, for $t_\ast\leq t\leq t_0+T$, we have
\begin{align*}
    \int_{t_\ast}^{t} \frac{\dot{V}(\zeta(\tau))}{V(\zeta(\tau))} \mathrm{d}\tau \!\leq\!
    \int_{t_\ast}^{t}-\frac{\epsilon_1\mu^{1+m}(\tau) }{\lambda_{\max}(P)}\mathrm{d}\tau 
    \!=\! - \frac{\epsilon_1 T [\mu^m(t)-\mu^{m}_\ast]}{m \lambda_{\max}( P)}
\end{align*}
where we used the fact that
\[
\frac{d}{d\tau}\mu^m(\tau) = \frac{m}{T}\mu^{1+m}(\tau), \quad \mu^{m}(t_\ast)=\mu_\ast^{m}.
\]
Since $\int_{t_\ast}^{t} \frac{\dot{V}(\zeta(\tau))}{V(\zeta(\tau))} \mathrm{d}\tau = \ln\left(\frac{V(\zeta(t))}{V(\zeta(t_\ast))}\right)$, we obtain
\begin{align*}
    V(\zeta(t)) \leq V(\zeta(t_\ast)) \exp \left( - \frac{\epsilon_1 T[\mu^m(t) - \mu^{m}_\ast]}{m \lambda_{\max}( P)}  \right).
\end{align*}
Using $\lambda_{\min}\!(P)\|\zeta(t)\|^2 \!\leq\! V(\zeta (t)\!) \!\leq\! \lambda_{\max}\!(P)\|\zeta(t)\|^2$, we obtain
\begin{equation}
\label{eq:zeta-bound}
\| \zeta(t) \|\! \leq \!\!\sqrt{\!\frac{\lambda_{\max}(\!P)}{\lambda_{\min} (\!P)}}\| \zeta(t_\ast) \|  \exp\!\!\left(\!
     - \frac{\epsilon_1 T[\mu^m(t) - \mu^{m}_\ast]}{2m \lambda_{\max}( P)} \!
    \right).
\end{equation}
The right-hand side of \eqref{eq:zeta-bound} can be written as
\begin{equation}\label{eq:klform}
\beta(r,s(t)) = \Omega r \exp(-\alpha s(t))
\end{equation}
where $r = \|\zeta(t_\ast) \|$, $s(t)=\mu^{m}(t)-\mu^{m}_\ast$, $\Omega = \sqrt{\frac{\lambda_{\max}(P)}{\lambda_{\min}(P)}}$, and $\alpha = \frac{\epsilon_1T}{2m\lambda_{\max}(P)}$.
For each fixed $s(t)$, the function $\beta(r,s(t))$ is strictly increasing in $r$ with $\beta(0,s(t))=0$. For each fixed $r$, $\beta(r,s(t))\to 0$ because $s(t)\to \infty$ and $\exp(-\alpha s(t)) \to 0$ as $t\to t_0+T$. Hence, $\beta$ qualifies as a class $\mathcal{KL}$ function, which by Definition~\ref{def:eFTGUAS} implies that the system \eqref{eq:transformederror} is E-FGUAS. 


Having established the bound for the transformed error $\zeta(t)$, we now consider the original error $e(t)=\overline \Gamma(t) \zeta(t)$, where we have $\| e(t)\|\leq \|\overline{\Gamma}(t)\|\|\zeta(t)\|$.
From \eqref{eq:zeta-bound}, we have
\[
\| e(t) \| \leq \Omega\|\overline{\Gamma}(t) \|\| \overline{\Gamma}^{-1}(t_\ast) \| \| e(t_\ast) \|  \exp\left(
     - \alpha[\mu^m(t) - \mu^{m}_\ast] 
    \right)
\]
where
\[
\|\overline{\Gamma}(t) \|\| \overline{\Gamma}^{-1}(t_\ast) \| = \frac{\mu^p(t)}{\mu_\ast^{1+m}}
\]
with $p=(1+m)\max(n_1,\dots,n_N)$. Therefore,
\[
\| e(t) \| \leq \frac{\Omega \|e(t_\ast) \|}{\mu_\ast^{1+m}}\exp(p\ln\mu(t) -\alpha[\mu^m(t) -\mu_\ast^m]).
\]
Since $\mu^m(t) = s(t) +\mu_\ast^m$, we can rewrite the right-hand side
\[
\beta_e(r_e,s) = \frac{\Omega r_e}{\mu_\ast^{1+m}}(s(t)+\mu^m_\ast)^{\frac{p}{m}}\exp({-\alpha s(t)})
\]
where $r_e = \| e(t_\ast) \|$. Again, for fixed $s(t)$, $\beta_e$ is strictly increasing in $r_e$ with $\beta_e(0, s(t)) = 0$. Moreover, as $t \to t_0 + T$, we have $s(t) \to \infty$ due to $\mu(t) \to \infty$.
Using L'H\^{o}pital's rule, we have that exponential decay dominates polynomial growth:
\[
\lim_{s \to \infty} (s + \mu_*^m)^{\frac{p}{m}} \exp({-\alpha s}) = 0
\]
implying $\beta_e(r_e, s(t)) \to 0$ as $t \to t_0 + T$. 
Hence, $\beta_e$ qualifies as a class $\mathcal{KL}$ function, which implies that the original error system \eqref{eq:errorthm} is E-FGUAS.
\end{proof}

\begin{remark}
The centralized prescribed-time observer in \cite{Adil2024} employs a proof strategy with similarities to our approach. 
However, the proof of \cite[Theorem 1]{Adil2024} contains several mathematical errors, as described in Appendix~\ref{append:mistakes-adil2024}. 
Similarly, a distributed prescribed-time observer in \cite{Chang2021} relies on a wrong lemma, which is also falsified in Appendix~\ref{append:mistakes-chang2021}.
These errors affect the validity of their stated results. 
We provide this observation to clarify the distinctions between our work and the existing literature, and to ensure mathematical rigor in this developing field.
\end{remark}

\section{Numerical Simulations}

To demonstrate the performance of the distributed observer, we consider the example from \cite{Yuanqing22}, which is a four-agent ($N =4$) system of the form \eqref{eq:sys} given by
\begin{align}\label{eq:examplesys}
\begin{aligned}
    \dot x_{11} & = x_{12}, ~~
    \dot x_{12} = -x_{11}+x_{12}(1-x^2_{11})\\
    \dot x_{21} & = -x_{21}-\cos(x_{11})\\
    \dot x_{31} & = x_{32}, ~~
    \dot x_{32} = -2x_{31} + \sin(x_{12})+x_{21}\\
   \dot x_{41}  & = - \frac{x_{41}}{1+x^{2}_{41}}+x_{31}\\
        y^{1}   &= x_{11}, ~ 
        y^{2}   = x_{21},~  
        y^{3}   = x_{31}, ~ 
        y^{4}     = x_{41}
   \end{aligned}
\end{align}
where $n_1 = n_3 = 2$ and $n_2 = n_4 = 1$.
The agents communicate over a strongly connected directed graph illustrated in Fig.~\ref{fig:ex-digraph}, whose adjacency and Laplacian matrices are
\begin{equation*}
\mathcal{A} =
\begin{bmatrix}
0 & 1 & 0 & 1\\
0 & 0 & 1 & 0\\
1 & 1 & 0 & 0\\
1 & 0 & 0 & 0
\end{bmatrix}, \quad
    \mathcal{L} =
    \begin{bmatrix}
        2 & -1 & 0 & -1\\
        0 & 1 & -1 & 0\\
        -1 & -1 & 2 & 0\\
        -1 & 0 & 0 & 1
    \end{bmatrix}.
\end{equation*}
The normalized positive left eigenvector of the Laplacian matrix (see Fact~\ref{lem:stronggraphs}) is computed to be $r= \begin{bmatrix}
    0.8 & 1.6 & 0.8 & 0.8
\end{bmatrix}$.

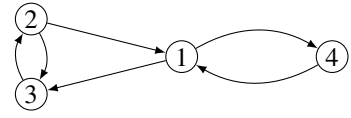
\begin{figure}[!t]
    \centering
    \begin{tikzpicture}
        \vertex (1) at (0,0) {1};
        \vertex (2) at (-2,0.5) {2};
        \vertex (3) at (-2,-0.5) {3};
        \vertex (4) at (2,0) {4};
        \path[-latex]
            (1) edge (3)
            (1) edge[bend left=30] (4)
            (2) edge (1)
            (2) edge[bend left=30] (3) 
            (3) edge[bend left=30] (2)
            (4) edge[bend left=30] (1)
        ;
    \end{tikzpicture}
    \caption{Directed graph for the example.}
    \label{fig:ex-digraph}
\end{figure}

The system described by \eqref{eq:examplesys} is not globally Lipschitz over $\mathbb{R}^n$. 
However, \eqref{eq:examplesys} is smooth and its trajectories are bounded inside a compact set $\mathcal X\subset\mathbb R^n$, i.e., it is forward complete within $\mathcal X$. 
Therefore, as also discussed in \cite{Yuanqing22}, the Lipschitz condition is satisfied over the system's state space $\mathcal{X}$ and our approach is valid for this example. 
The Lipschitz constant of the first subsystem is computed as $\gamma_1 \approx 6$ over $t\in[0,2]$. The remaining subsystems are globally Lipschitz, with $\gamma_2 =\gamma_4 = \sqrt{2}$, and $\gamma_3 = \sqrt{6}$. Therefore, from \eqref{eq:k_f-ineq}, we get $k_f \approx 26.2$. 

For the distributed observer, we set the prescribed convergence time $T=2$ and convergence parameter $m=2$. For $t_0=0$ and $t_\ast=1.98$, we obtain $\mu_{\ast} = 100$ from \eqref{eq:mu}. We find a feasible solution for the LMIs \eqref{eq:LMISnew} using {YALMIP with Sedumi solver}, where we obtain $\epsilon_1=0.0045$, $\epsilon_2=0.3694$, and $\epsilon_3=1.0046$. Other decision variables are obtained as 
\begin{align*}
    P_1 & = 
    \colsep{0.5pt}\begin{bmatrix}
    0.4646 &  -0.2276 &   0       &0         &0         &0\\
    -0.2276 &   0.3882 &   0       &0         &0         &0\\
    0      &   0      &   0.4673  &0         &0         &0\\
    0      &   0      &   0       &0.2126    &-0.0635    &0\\
    0      &   0      &   0       &-0.0635    &0.7717    &0\\
    0      &   0      &   0       &0         &0         &0.4502
    \end{bmatrix}, \\
    P_2 & =
      \colsep{2pt}\begin{bmatrix}
    0.2398   &-0.0573  & 0      &0      &0      &0\\
   -0.0573   & 0.5125  & 0      &0      &0      &0\\
    0        & 0       & 0.3638 &0      &0      &0\\
    0        & 0       & 0      &0.4363 &0.0555 &0\\
    0        & 0       & 0      &0.0555 &0.6759 &0\\
    0        & 0       & 0      &0      &0      &0.4299
      \end{bmatrix}, \\
    P_3 & = 
    \colsep{2pt}\begin{bmatrix}
    0.2843    &0.0246   & 0      &0        &0       &0\\
    0.0246    &0.6024   & 0      &0        &0       &0\\
    0         & 0       & 0.4340 &0        &0       &0\\
    0         & 0       & 0      &0.4660   &-0.2655 &0\\
    0         & 0       & 0      &-0.2655  &0.4462  &0\\
    0         & 0       & 0      &0        &0       &0.4698
    \end{bmatrix}, \\
    P_4 & =
    \colsep{0.5pt}\begin{bmatrix}
    0.3754   &-0.0413    &0       &0        &0        &0\\
   -0.0413   &0.4790     &0       &0        &0        &0\\
    0        &0          &0.5232  &0        &0        &0\\
    0        &0          &0       &0.1995   &-0.0479  &0\\
    0        &0          &0       &-0.0479  &0.9124   &0\\
    0        &0          &0       &0        &0        &0.4172
    \end{bmatrix}, \\
    Q_1 & =
    \begin{bmatrix}
            0.9319  &  1.6939 &  0 & 0 & 0
    \end{bmatrix}, \\
    Q_2 & =
    \colsep{5.75pt}\begin{bmatrix}
            0 & 0 & 0.4940 & 0 & 0 & 0
    \end{bmatrix}, \\
     Q_3 & =
    \colsep{3.5pt}\begin{bmatrix}
            0 & 0 & 0 & 1.5606  &   1.8822 & 0
    \end{bmatrix}, \\
    Q_4 & =
    \colsep{5.75pt}\begin{bmatrix}
            0 & 0 & 0 & 0 & 0 &  0.4922
    \end{bmatrix},
\end{align*}
and the local observer gains \eqref{eq:Lsnew} are
\begin{align*}
    L^{1} &= \begin{bmatrix} 5.8117 & 7.7697 & 0 & 0 & 0 & 0 \end{bmatrix}^\T, \\   
    L^{2} &= \colsep{7.2pt}\begin{bmatrix} 0 & 0 & 1.3578 & 0 & 0 & 0 \end{bmatrix}^\T, \\
    L^{3} &= \begin{bmatrix} 0 & 0 & 0 & 8.7002 & 9.3939 & 0 \end{bmatrix}^\T, \\
    L^{4} &= \colsep{7.2pt}\begin{bmatrix} 0 & 0 & 0 & 0 & 0 & 1.1797 \end{bmatrix}^\T.
\end{align*}
In the simulation, the initial state of the observed system is set as $x(t_0) = \begin{bmatrix}
        1 & 0 & -1 & 2 & 0 & -2
    \end{bmatrix}^\T$
and $z(t_0)=0_{nN}$.

All simulations are done with \texttt{ode15s} and the results are illustrated in Fig.~\ref{fig:allfigs}. The first plot presents the trajectories of the observed system. The next four plots display the trajectories of the local observer errors. It can be seen that the errors converge to zero for the prescribed time $T=2$. Finally, the last plot demonstrates convergence of the observer estimate $z^1_{12}(t)$ for different $m$ values for the unmeasured state $x_{12}(t)$. As $m$ increases, the estimate converges faster.

\begin{remark}[Implementation]
To mitigate numerical challenges arising from the diverging gains \(\Gamma,\mu\) near \(T\), practical implementation strategies proposed by \cite{Holloway19} include: (i) saturating (e.g., when \(\mu_{\max}=10^{10}\)) or disabling the gains once the estimation error reaches an acceptable threshold ; (ii) selecting a slightly larger convergence time \(T_{\delta}=T+\delta\) (\(\delta>0\)) to avoid singularities at \(t=T+t_0\)); and (iii) implementing a reset mechanism. Specifically, if at some time \(t_\text{stop}\) we have \(\mu(t_\text{stop})=\mu_{\max}\), the observer is reinitialized by setting \( t_0\leftarrow t_\text{stop} ~\text{and}~ z(t_0)\leftarrow z(t_\text{stop}).\)
\end{remark}

\begin{figure}[t!]
    \centering
    \includegraphics[width=\linewidth]{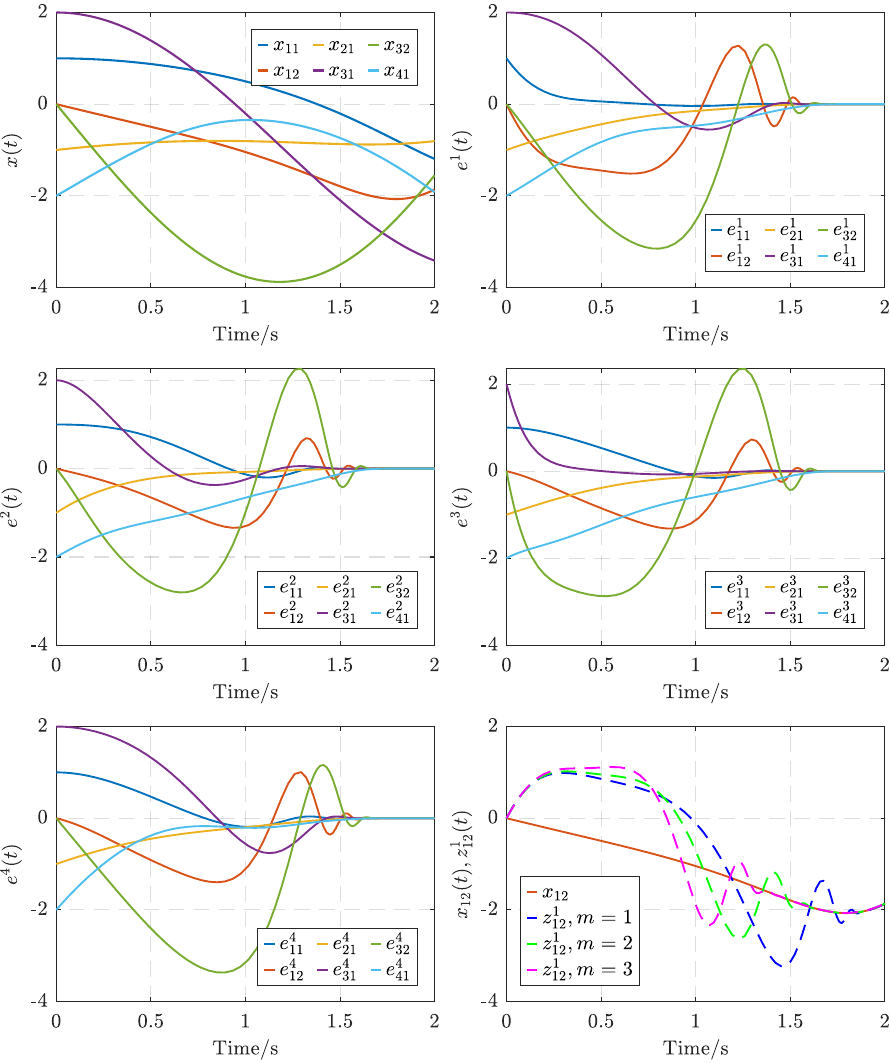}
    \caption{From top left to bottom right: the observed system, the errors of local observers, the estimates $z^{1}_{12}(t)$ for various $m$. 
    }
    \label{fig:allfigs}
\end{figure}

\section{Conclusion}

We presented a distributed prescribed-time observer design for nonlinear systems in block-triangular observable canonical form. 
Each local observer reconstructs the state within a user-specified time, regardless of initial conditions, by communicating over a strongly connected graph and utilizing time-varying, monotonically increasing gains.
The significance of this work lies in establishing the first theoretically sound framework for distributed prescribed-time observers for a class of nonlinear systems. Unlike fixed-time approaches where convergence time depends on observer parameters, our method allows the user to arbitrarily choose the desired convergence time for the distributed observer.

Several limitations remain to be addressed in future research. The current framework does not account for measurement noise, model uncertainties, communication delays, or time-varying network topologies. Additionally, exploring extensions to more general nonlinear system classes and investigating the effect of time-varying gains under noise measurements will be addressed next. Future work will also focus on robust versions that maintain performance guarantees under noisy and uncertain conditions.

\bibliographystyle{IEEEtran}
\bibliography{bibliography}

\appendix

\subsection{Proof of Lemma~\ref{lem:key-ineq}}
\label{append:tech-lemmas}
    
Note that $\Delta \Phi(t) = \col(\Delta \varphi^{1}(t),\ldots,\Delta \varphi^{N}(t)) =\col(\varphi(x)-\varphi(x-\Gamma\zeta^{1}),\ldots,\varphi(x)-\varphi(x-\Gamma\zeta^{N}))$.
Using \cite[Lemma 1]{Alessandri2013} and \cite[Eq. (23)]{Zemouche2019}, we have
\begin{equation}\label{eq:kfineqN1}
    \| \Gamma^{-1}_i(t) B_i \Delta \varphi^j_i(t)\| \leq n_i\gamma_i \|\zeta^j_i(t)\|
\end{equation}
where 
$\Delta\varphi^j_i(t) \coloneqq \varphi_i(\bar x_i(t)) - \varphi_i(\bar x_i(t) - \bar\Gamma_i(t) \bar\zeta^j_i(t))$
with $\bar\Gamma_i(t) = \diag(\Gamma_1(t),\dots,\Gamma_i(t))$, $\bar x_i(t)$ in \eqref{eq:Lipschitz}, and $\bar\zeta^j_i(t) = \col(\zeta^j_1(t),\dots,\zeta^j_i(t))$.
Further, note that $\overline{\Gamma}^{-1}(t) \overline B \Delta\Phi(t) =\col(v_1(t),\dots,v_N(t))$,
where
\[
v_j(t) = \col\left(\frac{B_1\Delta \varphi^{j}_1(t)}{\mu^{n_1(1+m)}(t)}, \dots, \frac{B_N\Delta \varphi^{j}_N(t)}{\mu^{n_N(1+m)}(t)}\right).
\]
Then, using \eqref{eq:kfineqN1}, it follows that
\begin{multline*}
    \left\|\overline{\Gamma}^{-1}(t)\overline{B}\Delta\Phi(t)\right\| = 
    \sqrt{\sum_{j=1}^N \sum_{i=1}^N \|\Gamma_i^{-1}(t)B_i\Delta\varphi_i^j(t)\|^2} \\
    \leq \sqrt{\sum_{j=1}^N \sum_{i=1}^N (n_i\gamma_i)^2 \|\zeta_i^j(t)\|^2}
    \leq \sqrt{N \sum_{i=1}^N  (n_i\gamma_i)^2}\cdot\|\zeta(t) \|.
\end{multline*}

\subsection{Analysis of Mathematical Inconsistencies in \cite{Adil2024}}
\label{append:mistakes-adil2024}

In this appendix, we provide a detailed analysis of several mathematical inconsistencies in \cite[Theorem 1 and 2]{Adil2024}. For clarity, we maintain the notation used in the original paper.

\subsubsection{Incorrect Application of Lipschitz Property}

Given the Lipschitz-like property of $f:\mathbb{R}^n\to\mathbb{R}$ stated in \cite{Adil2024}:
\begin{equation}
    \label{eq:lipschitz-ref14}
    |f(x_1,\dots,x_n) - f(\hat{x}_1,\dots,\hat{x}_n)| \leq \gamma_f \sum_{i=1}^n |x_i - \hat{x}_i|
\end{equation}
In the proof of \cite[Theorem 1]{Adil2024}, following Eq. (38), the authors claim:
\[
2\bar{e}^T(t) P \Gamma(t) B [f(x) - f(\hat{x})] \leq 2\bar{e}^T(t) P \Gamma(t) B [\gamma_f (x(t)-\hat{x}(t))]
\]
This expression contains two mathematical errors:
\begin{enumerate}
    \item The Lipschitz property in \eqref{eq:lipschitz-ref14} applies to the norm of the difference, not to the vector difference directly.
    
    \item The expression $B[\gamma_f(x(t)-\hat{x}(t))]$ is dimensionally inconsistent, as $B \in \mathbb{R}^{n \times 1}$ and $(x(t)-\hat{x}(t)) \in \mathbb{R}^{n \times 1}$, making their multiplication invalid.
\end{enumerate}
A mathematically valid approach to bounding the cross-term $2\bar{e}^T(t) P \Gamma(t) B [f(x) - f(\hat{x})]$ is available in \cite{Alessandri2013}. However, applying this correct approach would invalidate the results claimed in \cite{Adil2024}. It is often impossible to derive an observer design criterion for nonlinear systems that is independent of the Lipschitz constant of the nonlinearity. The criterion in \cite[Theorem 1]{Adil2024} appears to achieve this through incorrect mathematical derivations.

\subsubsection{Circular Reasoning and Inconsistent Parameter Requirements}

The proof in \cite[Theorem 1]{Adil2024} exhibits further mathematical inconsistencies:

\begin{enumerate}
    \item The authors select $t_e$ to satisfy Eq. (41), which depends on design parameters $P$, $K$, and $\bar{K}$ that have not yet been determined. This introduces circular reasoning in the proof.
    
    \item Eq. (41) constitutes an additional design requirement beyond conditions (33) and (34) stated in \cite[Theorem 1]{Adil2024}, yet it is not included in the theorem statement.
    
    \item The authors claim that a constant $\rho$ can be selected to satisfy Eq. (41), but $\rho$ is already constrained by Eq. (21) in \cite[Lemma 3]{Adil2024}, creating another conflict.
\end{enumerate}

These inconsistencies indicate that the proof structure relies on circular reasoning that does not support the claimed results.

\subsubsection{Incomplete Proof of Prescribed-Time Stability}

According to the definition in \cite{Adil2024}, a system $\dot{x} = f(t,x_{[t-\tau,t]})$ exhibits prescribed-time stability if, for all $t \in [0, T+\tau)$, the state $x(t)$ satisfies:
\begin{equation}
    \label{eq-pt-stability-ref14}
    |x(t)| \leq \beta(|x_0(t)|,\mu(t-\tau)-1)
\end{equation}
where $\beta$ is a class $\mathcal{KL}$ function and $\mu(t-\tau)$ increases to $\infty$ as $t \to T+\tau$.

However, the proof only establishes this bound for the interval $(t_e,T+\tau)$ rather than the entire $[0,T+\tau)$, thus failing to satisfy the definition of prescribed-time stability.

\subsubsection{Missing Design Criterion}
In the proof of \cite[Theorem 1]{Adil2024}, after Eq. (40), the authors introduce:
\[
\ell = \frac{\lambda_1}{\lambda_{\max}(P)} - \tau \left(2 + \frac{1+m}{T} \right)
\]
and assert that $\ell > 0$. This constitutes an additional constraint on the maximum eigenvalue of $P$, which must be included as a design criterion but is omitted from the theorem statement.

\subsubsection{Invalidation of Theorem 2}
Since \cite[Theorem 2]{Adil2024} relies on ``similar lines of argument" as the proof of \cite[Theorem 1]{Adil2024}, the mathematical errors identified above also invalidate the claims of \cite[Theorem 2]{Adil2024}.

\subsection{Analysis of Mathematical Inconsistencies in \cite{Chang2021}}
\label{append:mistakes-chang2021}
In this appendix, we contend that the proof of \cite[Lemma 5]{Chang2021} contains a fundamental mathematical error. Since this lemma is essential to their main result \cite[Theorem 1]{Chang2021}, this error invalidates \cite[Theorem 1]{Chang2021} as well.

In \cite[Appendix C]{Chang2021}, the authors decompose the state matrix $A$ as $A=\diag(A_{1i},A_{2i},A_{3i})$, where
\begin{align*}
    A_{1i} &= \begin{pmatrix}
        0_{(p_i-1)\times 1} & \tilde{A}_{1i} \\ 0 & 0_{1\times (p_i-1)}
    \end{pmatrix} \\
    A_{2i} &= \begin{pmatrix}
        0_{(\tau_i-1)\times 1} & \tilde{A}_{2i} \\ 0 & 0_{1\times(\tau_i-1)}
    \end{pmatrix} \\
    A_{3i} &= \begin{pmatrix}
        0_{(n-p_i-\tau_i+1)\times 1} & \tilde{A}_{3i} \\ 0 & 0_{1\times (n-p_i-\tau_i+1)}
    \end{pmatrix}
\end{align*}
with 
\begin{align*}
    \tilde{A}_{1i} &= \diag(\delta_1,\dots,\delta_{p_i-1}) \\
    \tilde{A}_{2i} &= \diag(\delta_{p_i},\dots,\delta_{p_i+\tau_i-2}) \\
    \tilde{A}_{3i} &= \diag(\delta_{p_i+\tau_i-1},\dots,\delta_{n-1}).
\end{align*}
However, this decomposition presents a dimensional inconsistency. Specifically:
\begin{align*}
&\dim(A_{1i}) + \dim(A_{2i}) + \dim(A_{3i}) \\ &\qquad = p_i + \tau_i + (n-p_i-\tau_i+2) \\
& \qquad = n+2 > \dim(A)=n
\end{align*}
where $\dim(\cdot)$ denotes the row dimension of a square matrix.

Even if we were to assume this decomposition is correct (perhaps due to a typographical error in \cite{Chang2021}), the pair $(A,C_i)$ would be unobservable. The block diagonal structure of $A$ results in a disconnected graph representation as shown in Figure~\ref{fig:graph}.

\begin{figure}
    \centering
    \includegraphics[width=0.45\textwidth]{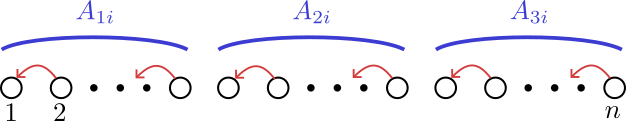}
    \caption{Graph representation of the system structure considered in \cite{Chang2021}.}
    \label{fig:graph}
\end{figure}

Due to this structure, there exist no edges connecting nodes from $A_{3i}$ to $A_{2i}$ or from $A_{2i}$ to $A_{1i}$. Furthermore, given that 
\begin{align}
C_i = \begin{pmatrix} 0_{1\times(p_i-1)}, \tilde{C}_i, 0_{1\times(n-p_i-\tau_i+1)} \end{pmatrix}
\end{align}
where $\tilde{C}_i = (1,0,\dots,0)\in\mathbb{R}^{1\times \tau_i}$, the observable subspace would consist solely of modes corresponding to $A_{2i}$, while all modes related to $A_{1i}$ and $A_{3i}$ would be unobservable. Moreover, the pair $(A,C_i)$ is not detectable since all eigenvalues of $A$ equal 0 (due to its upper triangular structure), meaning the unobservable modes are unstable. Consequently, the LMI in \cite[Lemma 5]{Chang2021} cannot be satisfied.

The correct decomposition of $A$ should be:
\begin{align*}
    A_{1i} &= \begin{pmatrix}
        0_{(p_i-1)\times 1} & \tilde{A}_{1i}
    \end{pmatrix}, \\
    A_{2i} &= \tilde{A}_{2i}, \\
    A_{3i} &= \begin{pmatrix}
        \tilde{A}_{3i} \\ 0_{1\times (n-p_i-\tau_i+1)}
    \end{pmatrix}
\end{align*}
Here, note that $A_{1i}$ and $A_{3i}$ are not square matrices.

Based on the observability of the pair $(A_{2i},\tilde{C}_i)$, one can verify that there exists a positive definite matrix $\tilde{P}_{\tau_i}$, a vector $L_{\tau_i}\in\mathbb{R}^{\tau_i}$, and scalars $0<\tilde{\alpha} < \tilde{\beta}$ such that the LMIs in \cite[Appendix C]{Chang2021}:
\begin{align*}
    \tilde{P}_{\tau_i}(A_{2i} + \tilde{L}_{\tau_i}\tilde{C}_i) + (A_{2i} + \tilde{L}_{\tau_i}\tilde{C}_i)^T \tilde{P}_{\tau_i} &\leq A_{2i} + A_{2i}^T - \gamma I_{\tau_i}  \\
    \tilde{\alpha} I_{\tau_i} \leq \tilde{P}_{\tau_i} H_i + H_i \tilde{P}_{\tau_i} &\leq \tilde{\beta} I_{\tau_i}
\end{align*}
are feasible. This part of their analysis is mathematically sound.

However, the authors then incorrectly extrapolate this result to the full system by choosing 
\begin{align*}
    P_i&=\diag(I_{p_i-1}, \tilde{P}_{\tau_i}, I_{n+1-\tau_i-p_i}) \\
    \bar{L}_i &= (0_{1\times(p_i-1)}, \tilde{L}_{\tau_i}^T, 0_{1\times(n+1-\tau_i-p_i)})
\end{align*}
and claiming that with appropriate selection of $0<\alpha<\beta$, the inequalities
\begin{align}
    & P_i (A+\bar{L}_i C_i) + (A+\bar{L}_i C_i)^T P_i \leq A + A^T - \gamma M_{\tau_i} \label{12-lmi-theorem} \\
    & \alpha I_n \leq P_i H + H P_i \leq \beta I_n \label{13-lmi-theorem}
\end{align}
are satisfied. This conclusion is mathematically invalid. The extrapolation would be valid if and only if $A_{1i}$, $A_{2i}$, and $A_{3i}$ were all square matrices, which would allow the matrix multiplications on the left-hand side of \eqref{12-lmi-theorem} to decouple. Since $A_{1i}$ and $A_{3i}$ are not square matrices, the relation $P_i A \neq \diag(A_{1i},\tilde{P}_{\tau_i} A_{2i},A_{3i})$ does not hold, contrary to what is required for their proof. Therefore, inequalities \eqref{12-lmi-theorem} and \eqref{13-lmi-theorem} cannot be established as claimed.

\end{document}